\documentclass[11pt]{article}

\usepackage{amsthm,amsmath,bm,amssymb,stmaryrd,fullpage,times}
\usepackage{tikz,url,program}

\newcommand{\ket}[1]{\mid\hspace*{-4pt}{#1}\rangle}

\newtheorem{definition}{{\bf Definition}}[section]
\newtheorem{remark}{{\bf Remark}}[section]
\newtheorem{proposition}{{\bf Proposition}}[section]
\newtheorem{corollary}{{\bf Corollary}}[section]

\makeatletter
\def\mop#1{\mathop{\operator@font {#1\null}}}
\makeatother
\def\bigO{{\mop{O}}}
\def\HW{{\mop{hw}}}

\def\Tr{{\mop{Tr}}}

\begin{document}
\title{Efficient quantum circuits for binary elliptic curve arithmetic:\\ reducing $T$-gate complexity}

\author{Brittanney Amento\\
Florida Atlantic University\\
Department of Mathematical Sciences\\
Boca Raton, FL 33431\\
{\tt bferoz@fau.edu}\and
Martin R{\"o}tteler\\
NEC Laboratories America\\
4 Independence Way, Suite 200\\
Princeton, NJ 08540, U.S.A.\\
{\tt mroetteler@nec-labs.com}
\and Rainer Steinwandt\\
Florida Atlantic University\\
Department of Mathematical Sciences\\
Boca Raton, FL 33431\\
{\tt rsteinwa@fau.edu}
}

\maketitle

\begin{abstract}
    Elliptic curves over finite fields ${\mathbb F}_{2^n}$ play a prominent role in modern cryptography. Published quantum algorithms dealing with such curves build on a short Weierstrass form in combination with affine or projective coordinates. In this paper we show that changing the curve representation allows a substantial reduction in the number of $T$-gates needed to implement the curve arithmetic. As a tool, we present a quantum circuit for computing multiplicative inverses in $\mathbb F_{2^n}$ in depth $\bigO(n\log n)$ using a polynomial basis representation, which may be of independent interest.
\end{abstract}

\section{Introduction}
Binary elliptic curves form an especially important family of groups for cryptographic applications, and the implementation of their addition law in a quantum circuit has been studied by a number of authors \cite{KaZa04,MMCP09b}. To the best of our knowledge, in all these discussions the representation used for elliptic curves is a short Weierstrass form in combination with affine or projective coordinates. While this is a natural choice, restricting to such representations does not exploit the available technical machinery---there is a substantial body of work on how to optimize elliptic curve arithmetic on classical hardware architectures (cf. \cite{EFDB12}), and one may hope that some of these classical results allow for simplification at the circuit level when implementing binary elliptic curve arithmetic in a quantum circuit, e.\,g., when trying to find discrete logarithms \cite{Sho97}. For an actual implementation, the number of $T$-gates needed to implement such a circuit is particularly of interest and it is desirable to keep this number as small as possible. The reason for this is that for most fault-tolerant quantum computing schemes, the implementation of $T$-gates is achieved via so-called magic state distillation~\cite{BK:2005,ASG:2009,Reichardt:2009}, a process which is costly in terms of physical resources required. For instance, in the case of the surface code, it is reasonable to assume that a single $T$-gate has a cost that is about $100$ times higher than a single CNOT~\cite{ASG:2009}. While minimizing the total number of $T$-gates is the prime objective of circuit synthesis at the logical level, the total depth of the computation when arranged as an alternation between $T$-gates and Clifford gates (the so-called ``$T$-depth'') is also an important parameter. It is desirable to keep the $T$-depth low by parallelizing $T$-gates as much as possible. 

\paragraph{Our contribution.} Below, we show how changing the curve representation can help to reduce the number of $T$-gates needed to implement elliptic curve arithmetic---and in addition help to reduce the circuit depth. The quantum circuit we present makes use of point addition formulae suggested by Higuchi and Takagi \cite{HiTa00} and can in particular be used to reduce the number of gates as well as the depth, in comparison to the use of ordinary projective coordinates (cf. \cite{MMCP09b}).

Some applications of elliptic curves may require unique representations of curve points (cf. \cite{MMCP09b}). When dealing with representations for fast arithmetic, deriving a unique point representation may involve an inversion in the underlying finite field. In a polynomial basis representation, a quantum implementation of the extended Euclidean algorithm can be used for this inversion, however the circuit has $\bigO(n^3)$ gates and quadratic depth \cite{KaZa04,MMCP09,MMCP09b}. For other field representations, an inversion algorithm with depth $\bigO(n\log n)$ and $\bigO(n^2\log n)$ gates has been proposed \cite{ARS12}. In order to compute unique point representations using a polynomial basis more efficiently, we adapt the approach from \cite{ARS12} to the polynomial basis setting. In this way we obtain the first published quantum circuit using a polynomial basis representation which can compute inverses in ${\mathbb F}_{2^n}^*$ in depth $\bigO(n\log n)$ with $\bigO(n^2\log n)$~gates.

\section{Fixing a finite field representation}\label{sec:fieldrep}
Fast addition formulae for points on an elliptic curve over a finite binary field ${\mathbb F}_{2^n}$ aim at reducing the number of (expensive) ${\mathbb F}_{2^n}$-operations. The following operations are of particular interest:
\begin{description}
   \item[Addition:] Given $\alpha,\beta\in{\mathbb F}_{2^n}$, compute their sum $\alpha+\beta$.
   \item[Multiplication:] Given $\alpha,\beta\in{\mathbb F}_{2^n}$, compute their product $\alpha\cdot\beta$.
   \item[Multiplication with a constant:] For a fixed non-zero constant $\gamma\in{\mathbb F}_{2^n}^*$, on input $\alpha\in{\mathbb F}_{2^n}$, compute $\gamma\cdot\alpha$. The value $\gamma$, for example, could be a coefficient in the defining equation of an elliptic curve.
   \item[Squaring:] Given $\alpha\in{\mathbb F}_{2^n}$, compute $\alpha^2$.
\end{description}
If one is interested in a unique representation of curve points, then the inversion of ${\mathbb F}_{2^n}$-elements also comes into play.
\begin{description}
\item[Inversion:] Given $\alpha\in{\mathbb F}_{2^n}^*$, find $\alpha^{-1}\in{\mathbb F}_{2^n}$.
\end{description}
The specific cost of each operation depends on how the field ${\mathbb F}_{2^n}$ is represented, and in the next two sections we look at three representations that have been considered in the literature on quantum circuits.
\subsection{Polynomial basis representation}\label{sec:polybasis}
In a polynomial basis representation, ${\mathbb F}_{2^n}$ is identified with a quotient ${\mathbb F}_2[x]/(f)$ where $f\in{\mathbb F}_2[x]$ is an irreducible polynomial of degree $n$. Each $\alpha\in{\mathbb F}_{2^n}$ is represented by the unique sequence $(\alpha_0,\dots,\alpha_{n-1})\in{\mathbb F}_2^n$ with $\alpha=\sum_{i=0}^{n-1}x^i+(f)$. In a quantum circuit, we store each coefficient $\alpha_i$ in a separate qubit. Quantum arithmetic in such a representation has been explored by a number of authors, including Beauregard et al. \cite{BBF03}, Kaye and Zalka \cite{KaZa04}, and Maslov et al. \cite{MMCP09b}. For each of the four basic tasks mentioned above, the exact implementation complexity varies depending on the particular choice of $f$ and efficient circuits are available:
\begin{description}
   \item[Addition:] As addition is defined coefficient-wise, $n$ CNOT gates are sufficient to derive the representation of $\alpha+\beta$ from those of $\alpha$ and $\beta$. These gates operate on disjoint wires and can be implemented in depth $1$. To realize an addition $\ket{\alpha}\ket{\beta}\ket{0}\mapsto\ket{\alpha}\ket{\beta}\ket{\alpha+\beta}$ where the sum is stored in a separate register, we can first add $\ket{\alpha}$ to $\ket{0}$, followed by adding $\ket{\beta}$, i.\,e., $2n$ CNOT gates and depth $2$ suffice. In particular, we do not need a single $T$-gate to implement ${\mathbb F}_{2^n}$-addition.
   \item[Multiplication:] Building on a classical Mastrovito multiplier \cite{Mas88,Mas91,MaHa04}, in \cite{MMCP09b} a linear depth quantum circuit is presented which derives the product $\alpha\cdot\beta$ from $\alpha,\beta\in{\mathbb F}_{2^n}$. This circuit requires $n^2$ Toffoli gates and $n^2-1$ CNOT gates. In particular, the $T$-gate complexity of a full ${\mathbb F}_{2^n}$-multiplication is quite substantial.\footnote{With a realization of \cite{AMMR12}, a Toffoli gate can be implemented without ancillae with seven $T$-gates (or $T^\dagger$-gates which we assume to have the same cost) in a circuit that has a $T$-depth of $3$.}
  \item[Multiplication with a constant:] Fix $\gamma\in {\mathbb F}_{2^n}^*$.  As multiplication with $\gamma$ is ${\mathbb F}_2$-linear, invoking a general multiplier is not necessary. Instead, we can realize multiplication by $\gamma$ as a matrix-vector multiplication with a suitable non-singular matrix $\Gamma$. An  $LUP$-decomposition of $\Gamma$ immediately yields a depth~$2n$ circuit that is comprised of no more than $n^2+n$ CNOTs. No Toffoli gates are needed.

  \item[Squaring:] No dedicated quantum circuit to implement the squaring map $\ket{\alpha}\ket{0}\mapsto\ket{\alpha}\ket{\alpha^2}$ has been proposed, but as squaring in ${\mathbb F}_{2^n}$ is ${\mathbb F}_2$-linear, it is enough to implement a matrix-vector multiplication in depth $2n$ using no more than $n\cdot(n+1)=n^2+n$ CNOTs. No Toffoli gates are needed.
\end{description}
Summarizing, among the above mentioned four basic operations, only the general multiplication involves $T$-gates, and their number unfortunately grows quadratic in the extension degree $n$. In cryptographic applications of elliptic curves, values of $n\ge 160$ are common. Hence, if we can save a general ${\mathbb F}_{2^n}$-multiplication at the expense of some additions, squarings or constant multiplications,  this can be of great value for the implementor of a quantum circuit.

So far, our discussion has ignored the inversion operation. The current literature offers only a circuit with a cubic number of gates and quadratic depth \cite{KaZa04}, making the two representations discussed in the next section seemingly more attractive for inversion. However, in Section~\ref{sec:nicepolyinversion} below, we will show that both the cubic gate complexity and the quadratic depth of this operation can be avoided by adapating the inversion technique used in \cite{ARS12} to the polynomial basis setting.

\subsection{Gaussian normal basis and ghost-bit basis representations}
Aiming for a more efficient inversion algorithm, in \cite{ARS12} two field representations are considered that differ from the polynomial basis representation just discussed: a \emph{ghost-bit basis} and a \emph{Gaussian normal basis} representation. For the purposes of this paper it is not necessary to discuss their technical details, and we restrict to looking at the cost of the relevant arithmetic operations:
\begin{description}
\item[Addition:] With a Gaussian normal basis, addition can be performed in the same way as with a polynomial basis. If a ghost-bit basis is available, elements in ${\mathbb F}_{2^n}$ are represented with $n+1$ bits, resulting again in two approaches for the addition. One approach is to add $\ket{\alpha}$ to $\ket{\beta}$ yielding one additional CNOT gate and a depth 1 circuit. The other approach is to add $\ket{\alpha}$ followed by $\ket{\beta}$ to $\ket{0}$ yielding two additional CNOT gates and a depth 2 circuit. Apart from these details, the addition operation is exactly the same as when using a polynomial basis representation.
\item[Multiplication:] If a ghost-bit basis is available, the multiplication $\alpha\cdot\beta$ of two field elements $\alpha, \beta \in {\mathbb F}_{2^n}$ can be realized in depth $n+1$ using $(n+1)^2$ Toffoli gates.

With a Gaussian normal basis of type~$t$, a quantum circuit of depth $(t+(t\bmod 2))\cdot n-1$ involving $(t + (t\bmod 2))\cdot n^2-n$ Toffoli gates is available for multiplying two elements in ${\mathbb F}_{2^n}$.
\item[Multiplication with a constant:] Choosing the matrix $\Gamma$ in accordance with the Gaussian normal basis or the ghost-bit basis, we can proceed as in the case of a polynomial basis. For a Gaussian normal basis this yields a circuit with $n^2+n$ CNOTs, and as a result of the extra bit used in a ghost-bit basis, for the latter we obtain a quantum circuit comprised of $(n+1)\cdot(n+2)=n^2+3n+2$ CNOT gates. No Toffolis are needed.
\item[Squaring:] This operation is for free since the square of a field element can be obtained by simply reading the coefficient vector in permuted order. Hence, no gates are required to implement the squaring operation and we require $n$ respectively $n+1$ CNOTs, all operating in parallel, to implement the map $\ket{\alpha}\ket{0}\mapsto\ket{\alpha}\ket{\alpha^2}$, .
\end{description}
Again, in terms of $T$-gate complexity, multiplication is the dominating operation, and the number of squaring operations in formulae for fast elliptic curve addition can be expected to be quite small. Consequently, using a polynomial basis representation looks preferable, even if the particular extension degree of interest affords a Gaussian normal basis of small type.

However, taking the computation of inverses into account---an operation that occurs in the derivation of a unique representation of a curve point---the situation seems to become more involved: In \cite{ARS12} an inversion circuit of depth $\bigO(n\log n)$ involving $\bigO(n^2\log n)$~gates has been presented. Compared to the quadratic depth and cubic gate complexity of the best published inversion circuit using a polynomial basis \cite{KaZa04}, this looks quite attractive. While \cite{KaZa04} builds on Euclid's algorithm, \cite{ARS12} builds on a classical technique by Itoh and Tsujii \cite{ItTs89}, which exploits that an efficient squaring algorithm is available. As mentioned, in the case of a Gaussian normal basis or a ghost-bit basis representation, the squaring operations in a quantum circuit are actually for free. To overcome the cubic gate complexity and quadratic depth requirements of inversion using a polynomial basis, the next section shows how to apply Itoh and Tsujii's algorithm with a polynomial basis.

\subsection{Itoh-Tsujii inversion with a polynomial basis representation}\label{sec:nicepolyinversion}
Let $\alpha\in{\mathbb F}_{2^n}$ be non-zero. As $\alpha^{-1}=\alpha^{2^n-2}$, the inverse of $\alpha$ can be computed through exponentiation. Itoh and Tsujii proposed a particularly efficient method to compute this power (see \cite{ItTs89,TaTa01,RHSCC05,Gua11}), if the squaring operation in ${\mathbb F}_{2^n}$ is inexpensive. The quantum circuits for inversion in \cite{ARS12} use exactly this technique when working with a field representation where squaring is just a permutation of the coefficient vector. Here we want to show that even with a polynomial basis, this approach is a very attractive alternative to Euclid's algorithm. To describe Itoh and Tsujii's approach, it is convenient to introduce some notation: for $i\ge 0$ we define  $\beta_i=\alpha^{2^i-1}$. Then our goal is to find $\alpha^{-1}=(\beta_{n-1})^2$ from $\beta_1=\alpha$. For this we exploit that
\begin{equation}\beta_{i+j}=\beta_{i}\cdot\beta_j^{2^i}\label{equ:betamagic}
\end{equation}for all $i,j\ge 0$. Writing $n-1=\sum_{i=1}^{\HW(n-1)}2^{k_i}$ with $\lfloor\log_2(n-1)\rfloor=k_1>k_2>\dots>k_{\HW(n-1)}\ge 0$, Itoh and Tsujii's strategy to find $\alpha^{-1}$ can be summarized in three steps:
\begin{enumerate}
   \item[(I)] Repeatedly apply Equation~\eqref{equ:betamagic} with $i=j$ to find all of $\beta_{2^0}, \beta_{2^1},\dots,\beta_{2^{k_1}}$.
   \item[(II)] Use Equation~\eqref{equ:betamagic} to find $\beta_{2^{k_1}+2^{k_2}}, \beta_{2^{k_1}+2^{k_2}+2^{k_3}}, \dots, \beta_{2^{k_1}+2^{k_2}+\dots+2^{k_{\HW(n-1)}}}(=\beta_{n-1})$.
   \item[(III)] Compute $\alpha^{-1}=(\beta_{n-1})^2$.
\end{enumerate}
Computing a value $\beta_{i+j}$ from given values $\beta_i$, $\beta_j$ by means of Equation~\eqref{equ:betamagic} involes one multiplication and an exponentiation by a fixed power of $2$. As mentioned in Section~\ref{sec:polybasis}, the multiplication can be implemented with $n^2$ Toffolis plus $n^2-1$ CNOT gates in a quantum circuit of depth~$\bigO(n)$. Differing from the situation in \cite{ARS12}, the exponentiation with $2^i$ is not for free, but as the map $\xi\mapsto\xi^{2^i}$ is ${\mathbb F}_2$-linear and bijective, we can implement it as a matrix-vector multiplication with a suitable non-singular $n\times n$ matrix having entries in ${\mathbb F}_2$. Thence, using an LUP-decomposition of this matrix, the needed exponentiation can be realized with $n^2+n$ CNOT gates in depth $2n$. Summarizing, we see that in a polynomial basis representation, one evaluation of Equation~\eqref{equ:betamagic} can be realized in depth $\bigO(n)$ using $n^2$ Toffolis and $2n^2+n-1$ CNOT gates.

Step~(I) in the above procedure requires $\lfloor \log_2(n-1)\rfloor-1$ evaluations of Equation~\eqref{equ:betamagic}, i.\,e., this step can be realized in depth $\bigO(n\log_2 n)$ by means of $(\lfloor \log_2(n-1)\rfloor-1)\cdot n^2$~Toffolis and $\bigO(n^2\log n)$ CNOT gates. In Step~(II), performing $\HW(n-1)-1$ evaluations of Equation~\eqref{equ:betamagic} sequentially, we obtain a depth of $\bigO(n\log n)$, involving $(\HW(n-1)-1)\cdot n^2$ Toffolis and $\bigO(n^2\log n)$ CNOT gates. Step~(III) is just a matrix-vector multiplication with a suitable non-singular $n\times n$ matrix, and using an LUP-decomposition of the latter, a quantum circuit with no more than $n^2+n$ CNOT gates can realize this squaring in depth $2n$.

To `uncompute' ancilla, we run the complete circuit---with exception of the final squaring---`backwards' and obtain the following:

\begin{proposition}
   In a polynomial basis representation, $\alpha^{-1}$, the inverse of an element $\alpha\in{\mathbb F}_{2^n}$, can be computed in depth $\bigO(n\log_2(n))$ using $2\cdot (\lfloor \log_2(n-1)\rfloor+ \HW(n-1)-2)\cdot n^2=\bigO(n^2\log n)$ Toffolis and $\bigO(n^2\log n)$ CNOT gates. This includes the cost for cleaning up ancillae.
\end{proposition}

\begin{remark}
Organizing the computation of $\beta_{n-1}$ in Step~(II) in a tree structure, the circuit depth for this step can be reduced to $\bigO(n\log\log n)$, but because of Step~(I), for the overall depth of the inverter we still obtain the bound $\bigO(n\log_2 n)$.
\end{remark}
Even though the squaring operation is not for free, in terms of $T$-gate complexity, this inverter seems quite competitive to the ones presented in \cite{ARS12} for ghost-bit and Gaussian normal basis representations. Thence, in the remainder of this paper we assume that a polynomial basis representation of the underlying field ${\mathbb F}_{2^n}$ is used.

\section{Binary elliptic curves}\label{sec:shortweierstrass}
Let $n\in{\mathbb N}$ be a positive integer and $\mathbb F_{2^n}$ a finite field of size $2^n$. For cryptographic applications, typical values are $n\in \{163,233,283\}$ \cite{FIPS1863}. 
Perhaps the most common representation of ordinary elliptic curves in characteristic $2$ is a \emph{short Weierstra{ss} form}, given by a polynomial in ${\mathbb F}_{2^n}[x,y]$:
\begin{equation}
   y^2+xy=x^3+a_2x^2+a_6\label{equ:weierstrass}
\end{equation}
Here $a_2,a_6\in {\mathbb F}_{2^n}$, with $a_6\ne 0$, and for practical purposes one often has $a_2\in\{0,1\}$ (cf. \cite{FIPS1863}). We write $${\mathrm E}_{a_2,a_6}({\mathbb F}_{2^n}):=\{(u,v)\in{\mathbb F}_{2^n}: v^2+uv=u^3+a_2u^2+a_6\}\cup\{{\mathcal O}\}$$
for the (${\mathbb F_{2^n}}$-rational points on the) elliptic curve given by Equation~\eqref{equ:weierstrass}. The point ${\mathcal O}\in {\mathrm E}_{a_2,a_6}({\mathbb F}_{2^n})$ corresponds to the `point at infinity.'\footnote{More technically, $\mathcal O$ is the unique point that is obtained when passing to the projective closure of ${\mathrm E}_{a_2,a_6}$.} Because of $a_6\ne 0$, we have $(0,0)\not\in {\mathrm E}_{a_2,a_6}({\mathbb F}_{2^n})$, suggesting $(0,0)\in{\mathbb F}_{2^n}^2$ as convenient representation of $\mathcal O$. Hence, each curve point can be naturally represented as a pair of two field elements (which fit into $2n$ qubits). The elliptic curve ${\mathrm E}_{a_2,a_6}({\mathbb F}_{2^n})$ is equipped with a natural group structure, where $\mathcal O$ serves as the identity. Namely, for $P_1=(x_1,y_1)$ and $P_2=(x_2,y_2)$, their sum $P_3=P_1+P_2$ can be computed by the procedure in Figure~\ref{fig:additionlaw}, which is taken from \cite{Sol98}.
\begin{figure}[htb]
\begin{center}
\begin{minipage}{0.6\textwidth}
\begin{program}
  \IF P_1={\mathcal O} \THEN |return|\ P_2\untab\untab
  \IF P_2={\mathcal O} \THEN |return|\ P_1\untab\untab
  \IF x_1=x_2 \THEN \IF y_1+y_2=x_2 \hspace*{4em}\rcomment{\# $P_1=-P_2$}
 \THEN |return|\ {\mathcal O}
                    \ELSE \lambda\gets x_2+y_2/x_2 \rcomment{\# $P_1=P_2$}
                          x_3\gets\lambda^2+\lambda+a_2
                          y_3\gets x_2^2+(\lambda+1)x_3\untab\untab
              \ELSE \lambda\gets (y_1+y_2)/(x_1+x_2) \rcomment{\# $P_1\ne\pm P_2$}
                                 x_3\gets\lambda^2+\lambda+x_1+x_2+a_2
                                 y_3\gets(x_2+x_3)\lambda+x_3+y_2\untab\untab
                                |return|\ (x_3,y_3)           
\end{program}
\end{minipage}
\end{center}
\caption{adding two points on the elliptic curve $y^2+xy=x^3+a_2x^2+a_6$}\label{fig:additionlaw}
\end{figure}


\subsection{Choosing a curve representation: the cost of adding a fixed point}\label{sec:howtoadd}
Before looking at the task of implementing a general point addition $P_1+P_2$, it is worthwhile to consider the special case when $P_1\ne\mathcal O\ne P_2$, $P_1\ne \pm P_2$, and $P_2$ is a fixed point. In a discrete logarithm computation as discussed in \cite{KaZa04,MMCP09b}, this is the only case needed, i.\,e., only the very last case of the addition law in Figure~\ref{fig:additionlaw} needs to be taken into account. Still, when using affine coordinates, the addition law involves an inversion in ${\mathbb F}_{2^n}$ and as indicated by the discussion in Section~\ref{sec:fieldrep}, this inversion operation is typically (much) more expensive to implement than addition or multiplication in ${\mathbb F}_{2^n}$. Therefore, relying on a projective formulation of the group law is a natural choice when designing quantum circuits. In projective coordinates, each $(x,y)\in {\mathrm E}_{a_2,a_6}({\mathbb F}_{2^n})\setminus\{{\mathcal O}\}$ is represented by a triple $(X,Y,Z)\in{\mathbb F}_{2^n}^{3}$ such that $X/Z=x$ and $Y/Z=y$, and $\mathcal O$ is represented by a triple $(0,Y,0)\in{\mathbb F}_{2^n}^{3}$ with $Y\ne 0$. These triples are only unique up to multiplication with a non-zero element in ${\mathbb F}_{2^n}$. Maslov et al. \cite{MMCP09b} exploit this freedom to restrict the number of of finite field inversion circuits in a discrete logarithm computation. In particular, they observe that as long as such a (non-unique) projective representation is sufficient, the addition of a constant curve point can be realized in linear depth.

To the best of our knowledge, no detailed (gate-level) analysis of how to add a fixed point on an elliptic curve has been published. Subsequently we note that---even with a clever implementation of projective coordinates---the $T$-gate complexity of such a quantum circuit can be reduced  substantially by passing to a different curve representation. As a welcome aside, it seems that simultaneously the circuit depth can be brought down.

\subsubsection{Mixed addition with projective coordinates}\label{sec:goodoldprojective} For the fixed point that is to be added, one can assume an affine representation is available leaving no need to handle a general `$Z$-coordinate' for this operand. So using projective coordinates, a natural (non-trivial) way to implement the addition of a fixed point is to apply the \emph{madd-2008-bl} formulae from  \cite{EFDB12}:
with the curve parameter $a_2$ as in Equation~\eqref{equ:weierstrass} these formulae derive a projective representation $(X_3, Y_3, Z_3)$ of $P_1+P_2$ with twelve ${\mathbb F}_{2^n}$-multiplications, three of them having one operand fixed (namely, one operand is $x_2$, $y_2$ or $a_2$), seven ${\mathbb F}_{2^n}$-additions, and one squaring.
$$\addtolength{\arraycolsep}{-1pt}
\begin{array}{l}
\begin{array}{lllllllll}
  A &=& Y_1 + Z_1\cdot y_2,&  B &=& X_1 + Z_1\cdot x_2,& AB &=& A+B,\\
  C &=& B^2,& E&=&B\cdot C,& F &=&(A\cdot AB+a_2\cdot C)\cdot Z_1+E,\\\hline
\end{array}\\
\begin{array}{lll}
   X_3 &=& B\cdot F,\\
   Y_3 &=& C\cdot(A\cdot X_1+B\cdot Y_1)+AB\cdot F,\\
   Z_3 &=& E\cdot Z_1.
\end{array}
\end{array}$$
Translating these formulae one by one immediately yields a quantum circuit in which the number of Toffolis, respectively $T$-gates, is determined by the nine general ${\mathbb F}_{2^n}$-multiplications. To reduce the circuit depth, we can try to parallelize some of the computations. Adding some CNOT gates to create `work copies' of intermediate results, we can enable parallelization without increasing the number of $T$-gates. To characterize the complexity of the resulting quantum circuit, we write $D_M(n)$ for the depth of an ${\mathbb F}_{2^n}$-multiplier $$\ket{\alpha}\ket{\beta}\ket{\xi}\mapsto\ket{\alpha}\ket{\beta}\ket{\xi+\alpha\beta},$$ and $G_M(n)$ for the number of gates required to implement such a multiplier. Further, we write $D_M^T(n)$ for the $T$-depth of an ${\mathbb F}_{2^n}$-multiplier and $G_M^T(n)$ for the number of $T$-gates required to implement such a multiplier. We assume that $D_M(n)$, $G_M(n)$, $D_M^T(n)$, and $G_M^T(n)$ include the cost for cleaning up ancillae.
Squaring operations and multiplications with a non-zero constant can be implemented with no more than $n^2+n$ CNOT gates in depth $2n$ each. As a functional composition of squarings and multiplications by a non-zero constant can be combined into a single invertible ${\mathbb F}_2$-linear map (through matrix multiplication), any fixed functional composition of squarings and non-zero constant multiplications can be implemented in depth $2n$ with $n^2+n$ CNOT gates as well.

\begin{proposition}\label{prop:naiveadd} The point addition $\ket{X_1}\ket{Y_1}\ket{Z_1}\ket{0}\ket{0}\ket{0}\longrightarrow\ket{X_1}\ket{Y_1}\ket{Z_1}\ket{X_3}\ket{Y_3}\ket{Z_3}$ can be implemented in overall depth $6D_M(n)$ plus $8n +\bigO(1)$ (the latter accounting for CNOT gates), and $T$-depth $6D_M^T(n)$. Further, a total of $15G_M(n)$ gates and $8n^2+\bigO(n)$ CNOT gates suffice. The total number of $T$-gates is $15G_M^T(n)$. This includes the cost for cleaning up ancillae.

Here $(X_3, Y_3, Z_3)$ is some projective representation of $P_1+P_2$ and $P_2\in {\mathrm E}_{a_2, a_6}({\mathbb F}_{2^n})$ a fixed point, represented with affine coordinates $(x_2, y_2)$. 
\end{proposition}
\begin{proof}
To implement the \emph{madd-2008-bl} formulae we can proceed as follows:

\begin{enumerate}
    \item Create a `work copy' $Z_1'$ of $Z_1$ using $n$ CNOT gates, all of which operate in parallel. Then compute $Z_1\cdot y_2$ and $Z_1'\cdot x_2$ in parallel and store these values in separate ($\ket{0}$-initialized) registers, using $2\cdot (n^2+n)$ CNOT gates and depth $2n$.
    \item Using $2n$ CNOT gates, all operating in parallel, add $Y_1$ to $Z_1\cdot y_2$ and add $X_1$ to $Z_1'\cdot x_2$, so that those registers now hold $A$ and $B$ respectively. Using $2n$ additional CNOT gates and increasing the circuit depth by $2$, we can also store $AB=A+B$ in a new ($\ket{0}$-initialized) register. 
		Moreover, using $2n$ CNOT gates, we can in constant depth provide `work copies' $A'$ of $A$ and $B'$ of $B$.

   \item Using $n^2+n$ CNOT gates, we can now compute $C=B^2$ in depth $2n$. If $a_2\ne 0$, with no more than $n^2+n$ additional CNOT gates we can in parallel determine $a_2\cdot(B')^2$.

   \item Using four multiplication circuits that operate in parallel, we can now compute $E=B\cdot C$, $A\cdot AB$, $A'\cdot X_1$ and $B'\cdot Y_1$ in depth $D_M(n)$, using $4\cdot G_M(n)$ gates.

\item Next, using $\le2n$ CNOT gates that operate in parallel we can add $A'\cdot X_1$ to $B'\cdot Y_1$ and---if $a_2\ne 0$---$A\cdot AB$ to $a_2\cdot (B')^2$.
\item With three general ${\mathbb F}_{2^n}$-multipliers we can now compute $(A\cdot AB+a_2\cdot (B')^2)\cdot Z_1'$, $C\cdot (A'\cdot X_1+B'\cdot Y_1)$, $Z_3=E\cdot Z_1$ and store these values in new registers. For this, depth $D_M(n)$ and $3\cdot G_M(n)$ gates suffice.
\item By adding $(A\cdot AB+a_2\cdot C)\cdot Z_1'$ to $E$ we obtain the value $F$ in depth $1$---involving $n$ CNOT gates. Increasing the depth by $1$ and adding $n$ more CNOT gates, we can also create a `work copy' $F'$ of $F$.
\item Invoking two more multiplication circuits, we can obtain $X_3=B\cdot F$ and $AB\cdot F'$ in depth $D_M(n)$ with $2\cdot G_M(n)$ gates.
\item Finally, adding $AB\cdot F'$ to $C\cdot (A'\cdot X_1+B'\cdot Y_1)$ yields $Y_3$, and this addition can be realized in depth~$1$ with $n$ CNOT gates.
\end{enumerate}
To clean up ancillae, the circuit is run backwards, excluding the final multiplications to compute $Z_3=E\cdot Z_1$, $X_3=B\cdot F$, the multiplication $C\cdot(A'\cdot X_1+B'\cdot Y_1)$, and the final addition to compute $Y_3$. This increases the overall depth by $3D_M(n)$ plus $4n+\bigO(1)$ (the latter accounting for CNOT gates), the $T$-depth by $3D_M^T(n)$, the gate count by an additional $6G_M(n)$ plus $4n^2+\bigO(n)$ (the latter accounting for CNOT gates), and the $T$-gate count by $6G_M^T(n)$.\qed
\end{proof}

\subsubsection{Mixed addition with a formula by Higuchi and Takagi}\label{sec:HigTak}
Building on earlier work by L\'opez and Dahab \cite{LoDa98}, in \cite{HiTa00} Higuchi and Takagi suggest a method to add points on an elliptic curve, which requires fewer multiplications than the \emph{madd-2008-bl} formulae we just discussed. Again, we consider the case of a point addition $P_1+P_2$ with $P_1\ne\pm P_2$ and $P_1\ne{\mathcal O}\ne P_2$, where $P_2$ is fixed. Instead of the usual projective coordinates $(X,Y,Z)$ with $x=X/Z$ and $y=Y/Z$ satisfying Equation~\eqref{equ:weierstrass}, Higuchi and Takagi choose a projective representation with $x=X/Z$ and $y=Y/Z^2$. The corresponding projective formulation of Equation~\eqref{equ:weierstrass} then becomes $$Y^2 + XYZ = X^3Z + a_2X^2Z^2 + a_6Z^4,$$
and the identity element $\mathcal O$ is represented by $(X,0,0)\in{\mathbb F}_{2^n}^3$ with $X\in{\mathbb F}_{2^n}^*$ arbitrary. For adding a curve point $P_1$ represented in these coordinates by $(X_1, Y_1, Z_1)\in{\mathbb F}_{2^n}^3$ to a fixed curve point $P_2$ given by affine coordinates $(x_2,y_2)\in{\mathbb F}_{2^n}^2$, ten ${\mathbb F}_{2^n}$-multiplications along with nine ${\mathbb F}_{2^n}$-additions and three squarings suffice. In two of the ten multiplications one operand is constant: 

$$\addtolength{\arraycolsep}{-1pt}
\begin{array}{l}
\begin{array}{lllllllll}
  A &=& x_2\cdot Z_1,&  B_1 &=& X_1^2,& B_2 &=& A^2,\\
  C &=& X_1+A,& D&=&B_1+B_2,& E &=&y_2\cdot Z_1^2,\\
  F &=& Y_1+E,& G&=&F\cdot C,\\\hline
\end{array}\\
\begin{array}{lll}
   Z_3 &=& Z_1\cdot D,\\
   X_3 &=& X_1\cdot(E+B_2)+A\cdot(Y_1+B_1),\\
   Y_3 &=& (X_1\cdot G+Y_1\cdot D)\cdot D + (G+Z_3)\cdot X_3.
\end{array}
\end{array}$$
Allowing an additional squaring, which does not affect the $T$-gate complexity, the formula for $Y_3$ can be rewritten as

\begin{equation}\label{equ:newy3}
  Y_3=X_1\cdot D\cdot G+Y_1\cdot D^2 + (G+Z_3)\cdot X_3.
\end{equation}
This latter formulation is helpful in deriving a quantum circuit with fewer $T$-gates and a lower depth than the one in Proposition~\ref{prop:naiveadd}:

\begin{proposition}\label{prop:cleveradd} The point addition 
$$\ket{X_1}\ket{Y_1}\ket{Z_1}\ket{0}\ket{0}\ket{0}\longrightarrow\ket{X_1}\ket{Y_1}\ket{Z_1}\ket{X_3}\ket{Y_3}\ket{Z_3}$$
can be implemented in overall depth $4D_M(n)$ plus $4n+\bigO(1)$ (the latter being CNOT gates), and $T$-depth $4D_M^T(n)$. Further, a total of $13G_M(n)$ gates and $8n^2+\bigO(n)$ CNOT gates suffice. The total number of $T$-gates is $13G_M^T(n)$. This includes the cost for cleaning up ancillae.

Here $(X_3, Y_3, Z_3)$ is some projective representation of $P_1+P_2$ as used by Higuchi and Takagi and $P_2$ a fixed curve point that is represented with affine coordinates $(x_2,y_2)$.
\end{proposition}
\begin{proof}
  To implement the point addition formulae by Higuchi and Takagi we can proceed as follows:
\begin{enumerate}
\item Using $3n$ CNOT gates, in depth $2$ we create `work copies' $X_1'$ of $X_1$ as well as $Z_1'$, $Z_1''$ and $Z_1'''$ of $Z_1$.
\item With no more than $4\cdot(n^2+n)$ CNOT gates, use the matrix-vector multiplications to compute $A=x_2\cdot Z_1$, $B_1=X_1^2$, $B_2=(x_2\cdot Z_1')^2$ and $E=y_2\cdot(Z_1'')^2$ which can be performed in parallel in depth $2n$. To be able to compute $D^2$, using $2\cdot(n^2+n)$ CNOT gates, we also compute in parallel $B_1^2=(X_1')^4$ and $B_2^2=(x_2\cdot Z_1''')^4$.

\item Using $\bigO(n)$ CNOT gates and constant depth we can now store $C=X_1+A$, $D=B_1+B_2$, a `work copy' $D'$ of $D$, and $F=Y_1+E$ in separate registers. Moreover, maintaining constant depth and with a linear number of CNOT gates, we can also store $E+B_2$, $Y_1+B_1$, and $D^2=B_1^2+B_2^2$; the latter three values will be used for computing $X_3$ and $Y_3$ respectively.
\item Now, using six general ${\mathbb F}_{2^n}$-multipliers, we can in parallel compute $G=F\cdot C$, $Z_3=Z_1\cdot D$, $X_1\cdot(E+B_2)$, $A\cdot(Y_1+B_1)$, $X_1'\cdot D'$, and $Y_1\cdot D^2$. For this, $6\cdot G_M(n)$ gates and depth $D_M(n)$ suffice.
\item At this point, $\bigO(n)$ CNOT gates and constant depth are adequate to compute $X_3=
X_1\cdot(E+B_2)+A\cdot(Y_1+B_1)$ and $G+Z_3$ and store these values in new registers.
\item With two more multipliers that operate in parallel, $(X_1'\cdot D')\cdot G$ and $(G+Z_3)\cdot X_3$ can be computed. Using $2\cdot G_M(n)$ gates, this can be accomplished in depth $D_M(n)$.
\item Finally, using $\bigO(n)$ CNOT gates and depth~$2$, with Equation~\eqref{equ:newy3} we can compute $Y_3=X_1\cdot D'\cdot G+Y_1\cdot D^2 + (G+Z_3)\cdot X_3$.
\end{enumerate}
To clean ancillae, we run the circuit backwards with the exception of the the final additions to compute $Y_3$ and $X_3$ and the multipliers to compute $Z_3=Z_1\cdot D$, $(G+Z_3)\cdot X_3$ and $A\cdot(Y_1+B_1)$. This increases the overall depth by $2D_M(n)$ plus $2n+\bigO(1)$ (the latter accounting for CNOT gates), the $T$-depth by $2D_M^T(n)$, the gate count by an additional $5G_M(n)$ plus $6n^2+\bigO(n)$ (the latter accounting for CNOT gates), and the $T$-gate count by $5G_M^T(n)$.\qed
\end{proof}
Comparing Proposition~\ref{prop:naiveadd} and Proposition~\ref{prop:cleveradd},  we see that passing from the usual projective representation to the one used by Higuchi and Takagi results in a significant saving in the total number of gates and $T$-gates while reducing the circuit depth and $T$-depth. Thence, replacing the usual projective addition in the quadratic depth solution for the discrete logarithm problem in \cite{MMCP09b} with the addition discussed in this section is an attractive implementation option.

\subsection{Implementing a general point addition using Edwards curves}
In view of the case distinctions in the addition law in Figure~\ref{fig:additionlaw}, implementing a quantum circuit that properly handles all cases of a point addition appears to be a somewhat burdensome task: in addition to the `generic case' $P_1\ne\pm P_2$ (with $P_2$ not being fixed) and $P_1\ne{\mathcal O}\ne P_2$, we have to implement a doubling formula ($P_1=P_2$), making sure that the identity element is handled properly ($P_1=-P_2$, $P_1={\mathcal O}$ or $P_2={\mathcal O}$). It is important to note here that testing the branching conditions in Figure~\ref{fig:additionlaw} comes at a certain cost when working with inversion-free arithmetic as just discussed. With projective coordinates as described in Section~\ref{sec:goodoldprojective}, let $(X_1,Y_1,Z_1)\in{\mathbb F}_{2^n}^3$ and $(X_2,Y_2,Z_2)\in{\mathbb F}_{2^n}^3$ be representations of two curve points $P_1$, $P_2$ different from the identity. Checking if these two points satisfy $$\underbrace{X_1/Z_1}_{x_1}=\underbrace{X_2/Z_2}_{x_2}\ (\iff X_1Z_2=X_2Z_1)$$ requires two ${\mathbb F}_{2^n}$-multiplications---not taking into account additional gates that may be needed to clean up ancillae. 

Working with a different representation of elliptic curves offers an elegant alternative to dealing with the case distinctions in Figure~\ref{fig:additionlaw}:
In \cite{BLF08}, Bernstein et al. discuss a representation of ordinary elliptic curves over ${\mathbb F}_{2^n}$ which affords a \emph{complete} addition law, i.\,e., the addition of any two curve points is handled with the very same formula. For $n\ge 3$ (which is especially safe to assume in cryptographic applications), each ordinary elliptic curve is birationally equivalent to such a \emph{complete binary Edwards curve} \cite{BLF08}.
\begin{definition}[Complete binary Edwards curve]
   Let $d_1, d_2\in{\mathbb F}_{2^n}$ with $\Tr(d_2)=1$. Then the \emph{complete binary Edwards curve with coefficients $d_1$ and $d_2$} is the affine curve defined by
$$d_1(x+y)+d_2(x^2+y^2)=xy+xy(x+y)+x^2y^2.$$
We will write ${\mathrm E}_{{\mathrm B},d_1, d_2}({\mathbb F}_{2^n})$ for the set of (${\mathbb F}_{2^n}$-rational) points on this curve.
\end{definition}
The identity element of a complete binary Edwards curve is $(0,0)\in {\mathrm E}_{{\mathrm B},d_1, d_2}({\mathbb F}_{2^n})$, and for \emph{any} two points $P_1=(x_1,y_1)$ and $P_2=(x_2,y_2)$ in  ${\mathrm E}_{{\mathrm B},d_1, d_2}({\mathbb F}_{2^n})$, their sum is $P_3=(x_3, y_3)$ with
\begin{eqnarray*}
   x_3&=&\frac{d_1(x_1+x_2)+d_2(x_1+y_1)(x_2+y_2)+(x_1+x_1^2)(x_2(y_1+y_2+1)+y_1y_2)}{d_1+(x_1+x_1^2)(x_2+y_2)}\text{ and}\\
   y_3&=&\frac{d_1(y_1+y_2)+d_2(x_1+y_1)(x_2+y_2)+(y_1+y_1^2)(y_2(x_1+x_2+1)+x_1x_2)}{d_1+(y_1+y_1^2)(x_2+y_2)}.
\end{eqnarray*}
Similar to working with a short Weierstrass form, one can pass to projective coordinates to avoid costly inversions. In \cite{BLF08} an explicit addition formula is given to compute a representation $(X_3,Y_3,Z_3)$ of the sum of two points on a complete binary Edwards curve, represented projectively as $(X_1,Y_1,Z_1)$ and $(X_2,Y_2,Z_2)$. The formula involves 21 general multiplications in ${\mathbb F}_{2^n}$, three multiplications by the parameter $d_1$, one multiplication by the parameter $d_2$, $15$ additions of ${\mathbb F}_{2^n}$-elements, and one squaring:
$$\addtolength{\arraycolsep}{-1pt}
\begin{array}{l}
\begin{array}{llllllllllll}
  W_1 &=& X_1 + Y_1,  & W_2 &=& X_2 + Y_2,& A &=& X_1\cdot (X_1 + Z_1),& B& =& Y_1 \cdot (Y_1 + Z_1),\\
  C   &=&Z_1\cdot Z_2,& D  &=& W_2\cdot Z_2,& E&=& d_1C^2,& H &=& (d_1Z_2 + d_2W_2)\cdot W_1 \cdot C,\\
  I &=& d_1Z_1\cdot C,& U &=& E + A\cdot D,& V &=& E + B\cdot D,& S &=& U\cdot V,
\end{array}\\
\begin{array}{lll}
   X_3 &=& S\cdot Y_1 + (H + X_2\cdot(I + A\cdot(Y_2 + Z_2)))\cdot V\cdot Z_1,\\
   Y_3 &=& S\cdot X_1 + (H + Y_2\cdot(I + B\cdot(X_2 + Z_2)))\cdot U\cdot Z_1,\\
   Z_3 &=& S\cdot Z_1.
\end{array}
\end{array}$$
These formulae can be translated into a quantum circuit for adding arbitrary (variable) curve points:

\begin{proposition}
Denote by $(X_1, Y_1, Z_1)$ and $(X_2, Y_2, Z_2)$ projective representations of two (not necessarily distinct) points $P_1, P_2\in E_{\mathrm{B}, d_1, d_2}$.
Then the point addition $$\ket{X_1}\ket{Y_1}\ket{Z_1}\ket{X_2}\ket{Y_2}\ket{Z_2}\ket{0}\ket{0}\ket{0}\longrightarrow\ket{X_1}\ket{Y_1}\ket{Z_1}\ket{X_2}\ket{Y_2}\ket{Z_2}\ket{X_3}\ket{Y_3}\ket{Z_3}$$ can be implemented in overall depth $5D_M(n)+4\max(D_M(n),2n)+\bigO(1)$, where the argument $2n$ of $\max(\cdot)$ as well as the $\bigO(1)$ reflect CNOT gates only, and $T$-depth $9D_M^T(n)$. Further, a total of $39G_M(n)$ plus $8n^2+\bigO(n)$ CNOT gates suffice. The total number of $T$-gates is $39G_M^T(n)$. At this, $(X_3,Y_3,Z_3)$ is a projective representation of $P_1+P_2$. This includes the cost for cleaning up ancillae.
\end{proposition}
\begin{proof}
   To implement the above addition formulae, we proceed as follows:
\begin{enumerate}
\item Compute in parallel the values $W_1, W_2$ as well as $X_1+Z_1$ and $Y_1+Z_1$, $Y_2+Z_2$, and $X_2+Z_2$  from the input values $X_1, Y_1, Z_1, X_2, Y_2, Z_2$---this can be done in constant depth using $\bigO(n)$ CNOT gates. In addition we use (depth~$1$) additions to $\ket{0}$ to create `work copies' $W_2'$ of $W_2$, $Z_1'$ of $Z_1$, and $Z_2', Z_2''$ of $Z_2$ using $3n$ CNOT gates.
\item Using four general ${\mathbb F}_{2^n}$-multipliers and two matrix vector multiplications, compute in parallel the values $A$, $B$, $C$, $D=W_2\cdot Z_2'$, along with $d_1Z_2''$ and $d_2W_2'$. As all involved multipliers operate on disjoint sets of wires, this can be done in depth $\max(D_M(n),2n)$ using no more than $4G_M(n)$ plus $2\cdot(n^2+n)$ gates (the latter accounting for CNOT gates).
\item Compute (in preparation for computing $H$) the value $d_1Z_2''+d_2W_2'$ and create `work copies' $A'$ of $A$, $B'$ of $B$, $C'$ of $C$, and $D'$ of $D$ using $\bigO(n)$ CNOT gates and constant depth.

\item Using five general ${\mathbb F}_{2^n}$-multipliers and two matrix vector multiplications, compute in parallel the values $E=d_1C^2$, $W_1\cdot C'$, $A\cdot D$, $B\cdot D'$, $A'\cdot(Y_2+Z_2)$, $B'\cdot(X_2+Z_2)$ and $d_1Z_1$. This can be done in depth $\max(D_M(n),2n)$ with no more than $5G_M(n)$ plus $2\cdot(n^2+n)$ gates (the latter accounting for CNOT gates).
\item Compute $U$ and $V$ and create `work copies' $U'$ of $U$ and $V'$ of $V$ in constant depth using $\bigO(n)$ CNOT gates.
\item Using five general ${\mathbb F}_{2^n}$-multipliers, find $H$, $I$, $S$, $U'Z'_1$ and $V'Z_1$ using $5G_M(n)$ gates in depth $D_M(n)$.
\item Compute $I+A\cdot(Y_2+Z_2)$ and $I+B'\cdot(X_2+Z_2)$ in constant depth using $\bigO(n)$ CNOT gates. Moreover, generate a `work copy' $S'$ of $S$ using $n$ CNOT gates and maintaining constant depth.
\item Using four general ${\mathbb F}_{2^n}$-multipliers, compute in parallel $X_2\cdot (I+A\cdot(Y_2+Z_2))$ and $Y_2\cdot(I+B(\cdot X_2+Z_2))$, $SX_1$ and $S'Y_1$, in depth $D_M(n)$ using $4G_M(n)$ gates.
\item Involving $\bigO(n)$ CNOT gates, compute $H+X_2\cdot (I+A\cdot(Y_2+Z_2))$ and $H+Y_2\cdot(I+B\cdot(X_2+Z_2))$ in depth~$2$.
\item Multiply $H+X_2\cdot (I+A\cdot(Y_2+Z_2))$ with $V'Z_1$, $H+Y_2\cdot(I+B\cdot(X_2+Z_2))$ with $U'Z'_1$, and compute $Z_3=S\cdot Z_1$. This can be done using $3G_M(n)$ gates in depth $D_M(n)$.
\item Compute $X_3$ by adding $S'Y_1$ to $(H+X_2\cdot (I+A\cdot(Y_2+Z_2)))\cdot V'Z_1$ and $Y_3$ by adding $SX_1$ to $(H+Y_2\cdot(I+B\cdot (X_2+Z_2)))\cdot U'Z'_1$ in depth~$1$ using $\bigO(n)$ CNOT gates.
\end{enumerate}
The above circuit has depth $3D_M(n)+2\max(D_M(n),2n)+\bigO(1)$ with the argument $2n$ of $\max(\cdot)$ as well as the $\bigO(1)$ originating in CNOT gates. The number of gates is bounded by $21G_M(n)$ plus $4n^2+\bigO(n)$ CNOTs. `Uncomputing' auxiliary qubits by running the circuit backwards---with the exception of the multiplications $Z_3=S\cdot Z_1$, $H+Y_2\cdot(I+B'\cdot (X_2+Z_2))\cdot U'Z'_1$, $H+X_2\cdot(I+A\cdot (Y_2+Z_2))\cdot V'Z_1$, and the final additions to compute $X_3$ and $Y_3$---yields the desired bound.\qed
\end{proof}

Making use of the (linear-depth and polynomial-size) multiplication circuits in \cite{ARS12}, for asymptotic purposes  we obtain the following corollary from the above proposition.

\begin{corollary}
Two points on an Edwards curve in projective representation can be added in linear depth with a polynomial-size quantum circuit.
\end{corollary}
\begin{proof}
  This follows immediately from the multiplier architectures described in \cite{ARS12}, which have linear depth and involve only a polynomial number of gates.\qed
\end{proof}

\section{Conclusion}
The circuits for binary elliptic curve arithmetic we have presented here are most likely not `optimal' yet, but they give ample evidence that incorporating results from the classic elliptic curve literature in quantum circuit design is worthwhile: it is possible to bring down the number of gates and $T$-gates that need to be protected against errors and it is possible to reduce the overall circuit depth and $T$-depth. We hope that our results stimulate follow-up work on the design of efficient quantum circuits for elliptic curve arithmetic---including the case of fields of odd characteristic. For adequately evaluting the cryptanalytic potential of quantum computers, this appears to be a fruitful and important research avenue.

\section*{Acknowledgments}
BA and RS acknowledge support by NSF grant No.~1049296 \emph{(Small-scale Quantum Circuits with Applications in Cryptanalysis)}.
\bibliographystyle{plain}
\bibliography{GF2}
\end{document}